\documentclass[%
 reprint,
 amsmath,amssymb,
 aps,
 pra,
longbibliography
]{revtex4-1}
\usepackage{soul}
\usepackage{graphicx}
\usepackage{dcolumn}
\usepackage{bm}
\usepackage{braket}
\usepackage{amsmath}
\usepackage{bbold}
\usepackage{xcolor}
\usepackage{amsthm}
\usepackage{physics}
\usepackage{enumitem}
\usepackage{centernot}
\newtheorem{theorem}{Theorem}
\newtheorem{proposition}{Proposition}

\theoremstyle{definition}
\newtheorem{definition}{Definition}
\theoremstyle{remark}
\newtheorem{remark}{Remark}

\renewcommand{\arraystretch}{1.7}

\begin{document}

\preprint{APS/123-QED}

\title{The asymptotic emergence of the Second Law for a repeated charging process}

\author{Marcin {\L}obejko, Pawe{\l} Mazurek, Micha{\l} Horodecki}
\affiliation{International Centre for Theory of Quantum Technologies, University of Gda\'nsk, 80-308 Gda\'nsk, Poland }

\date{\today}

\begin{abstract}
In one of its versions, the Second Law states: ``It is impossible to construct an engine which will work in a complete cycle, and produce no effect except the raising of a weight and cooling of a heat reservoir''. While the Second Law is considered as one of the most robust laws of Nature, it is still challenging how to interpret it in a fully quantum domain. Here we unpack the true meaning of the ``cyclicity'' and formulate the Second Law for a generic quantum battery via its asymptotic properties of a charging process rather than in terms of a single cycle. As a paradigm, we propose a machine consisting of a battery that repeatedly interacts with identically prepared 
systems. We then propose the Second Law in the form: The ergotropy of the battery may increase {\it indefinitely} if and only if systems are in a non-passive state. One of the most interesting features of this new formulation is the appearance of the passive states that naturally generalize the notion of the heat bath. In this paper, we provide a handful of results that supports this formulation for diagonal systems. Interestingly, our methodology meets a well-known theory of Markov chains, according to which we classify the general charging processes based on the passivity/non-passivity of charging systems. In particular, the adopted mathematics allows us to distinguish a subtle asymptotic difference  between the indefinite increase of the battery's energy (induced by the maximally mixed states) and of ergotropy (induced by the non-passive states) in terms of the so-called {\it null-recurrent} versus {\it transient} Markov chains. 
\end{abstract}

\maketitle

\onecolumngrid

\section{Introduction}

The Second Law is believed to be one of the most stable principles of Nature. In Planck's version, the Second Law sounds as follows: \emph{``It is impossible to construct an engine which will work in a complete cycle, and produce no effect except the raising of a weight and cooling of a heat reservoir''}. In other words, in macroscopic reality, the Second Law is the statement of something impossible; i.e., it is impossible to accumulate any potential energy (like gravitational or electrical) in a cyclic process, having the charging resource at equilibrium. 

Commonly, by charging, we mean the process of increasing the energy of some energy-storage device (i.e., a battery). This intuition, however, is not precise since the battery is the storage of the ``potential to work'' rather than the energy itself. It is understandable if one imagines that the battery may increase its energy through contact with a heat bath (i.e., via thermalization), which obviously violates Planck's statement. This is the reason why in thermodynamics we go beyond the average energy and replace it by the notion of \emph{free energy} or (more generally) \emph{ergotropy}. By definition, the ergotropy is the maximum amount by which the system is able to lower its energy in a cyclic, unitary process \cite{Allahverdyan2004}. A new class of the so-called \emph{passive states} with no ergotropy is introduced accordingly. More interestingly, Pusz and Woronowicz \cite{Pusz1978} have shown that a subclass of passive states (called \emph{completely passive states}), unable to decrease energy even if an arbitrary number of their copies is present, is uniquely given by the Gibbs states, i.e., the states at equilibrium. Consequently, one may claim the following is  Planck's formulation of the Second Law in a quantum domain: {\it The state of heat reservoir is completely passive state}. However, the Pusz-Woronowicz formulation does not introduce the battery explicitly, such that it only applies when the battery is a macroscopic object and, thus, is described by classical physics.

There is a question, therefore, whether one can extend the formulation of the Second Law to include an explicit energy-storage device. For microscopic batteries, quantum mechanics starts to play a role and the evolution of the system is not unitary anymore due to quantum back reaction. Consequently, this fully quantum formulation raises a big problem with the Second Law since the passivity argument cannot be used anymore. As an extreme example, let us bring attention to the recent result \cite{Tanmoy2022} showing that the battery (e.g., harmonic oscillator) coupled to the heat bath (i.e., completely passive state), in principle, may increase not only its energy but even the ergotropy (although still constrained by the fundamental bound expressed in terms of the free energy).  

So far, to remove this paradox, symmetries have been imposed. Namely, the battery was usually considered as an infinite weight, and the allowed operations were imposed to be translationally invariant \cite{Brunner2012, Skrzypczyk2014, Alhambra2016, Lobejko2020, Lobejko2021, Lobejko2022}. In such a setup, the battery's energy cannot be increased by interaction with the passive state, and the Second Law is saved. However, the above scenario is very unnatural - first of all, the unbounded from below the Hamiltonian of the battery is unphysical. Most importantly, the translational invariance of the operations is a very unnatural and strong condition. The Second Law should hold in some generic conditions rather than just apply to some artificial, highly symmetric scenario. Note that even the modern formulation of the Second Law expressed in the form of famous fluctuation relations \cite{Jarzynski1997, Campisi2011, Esposito2009} implicitly assumes weight as a work reservoir since it operates with differences of energy only \cite{Bartosik2021}. 

A critical ingredient of Planck's formulation is cyclicity - it is demanded that the only change is raising the weight and decreasing the energy of the bath. If the energy extraction is done through the working body, this would mean that the final state of the working body must be returned to its initial state (e.g., if we expand gas, we can raise the weight, but to complete the cycle, we have to compress the gas, which will return it to the initial state). However, in the presence of a microscopic battery, such cyclicity is not enough, as one can (ignoring even the presence of a working body) run such an interaction between battery and bath so that not only energy but also ergotropy will increase (reported in \cite{Tanmoy2022}, as we mentioned before). So even though the cycle was completed, the work was extracted against the formulation of the Second Law. 

Nevertheless, one may argue that, in spirit, Planck's formulation still holds since the battery state has to degrade in some sense after charging. In other words, an increase of ergotropy is possible due to the initial ``negentropy'' of the battery, while the weight from Planck's formulation was implicitly assumed not to suck entropy (which has also been adapted in the translational symmetry). It is, however, very hard to state the Second Law more precisely in terms of a single cycle for a generic battery since we cannot impose that battery's state will not change - it has to be charged. 

In this paper, we propose to resolve this problem by shifting the paradigm away from a single completed cycle to the asymptotic behavior of the charging process and formulate the Second Law in terms of its asymptotic behavior. Namely, we investigate what happens if we do not stop after the first cycle, but try to continue the charging. Then it is intuitive that the mentioned degradation of the battery's state should affect the charging process to such an extent that it reaches a steady state, making indefinite charging impossible. For definiteness, we consider many identically prepared systems and allow them to interact one by one with the battery. Finally, we postulate the Second Law as follows:\\

{\it The ergotropy of the battery may increase indefinitely if and only if the systems are in a non-passive state.} \\ \\
We stress that this formulation generalizes Planck's statement beyond heat reservoirs to the passive states, where the one-by-one machine-like setup is crucial since, otherwise, one can pair the systems to induce a non-passivity. 

\begin{figure}[t]
    \includegraphics[width=0.5\textwidth]{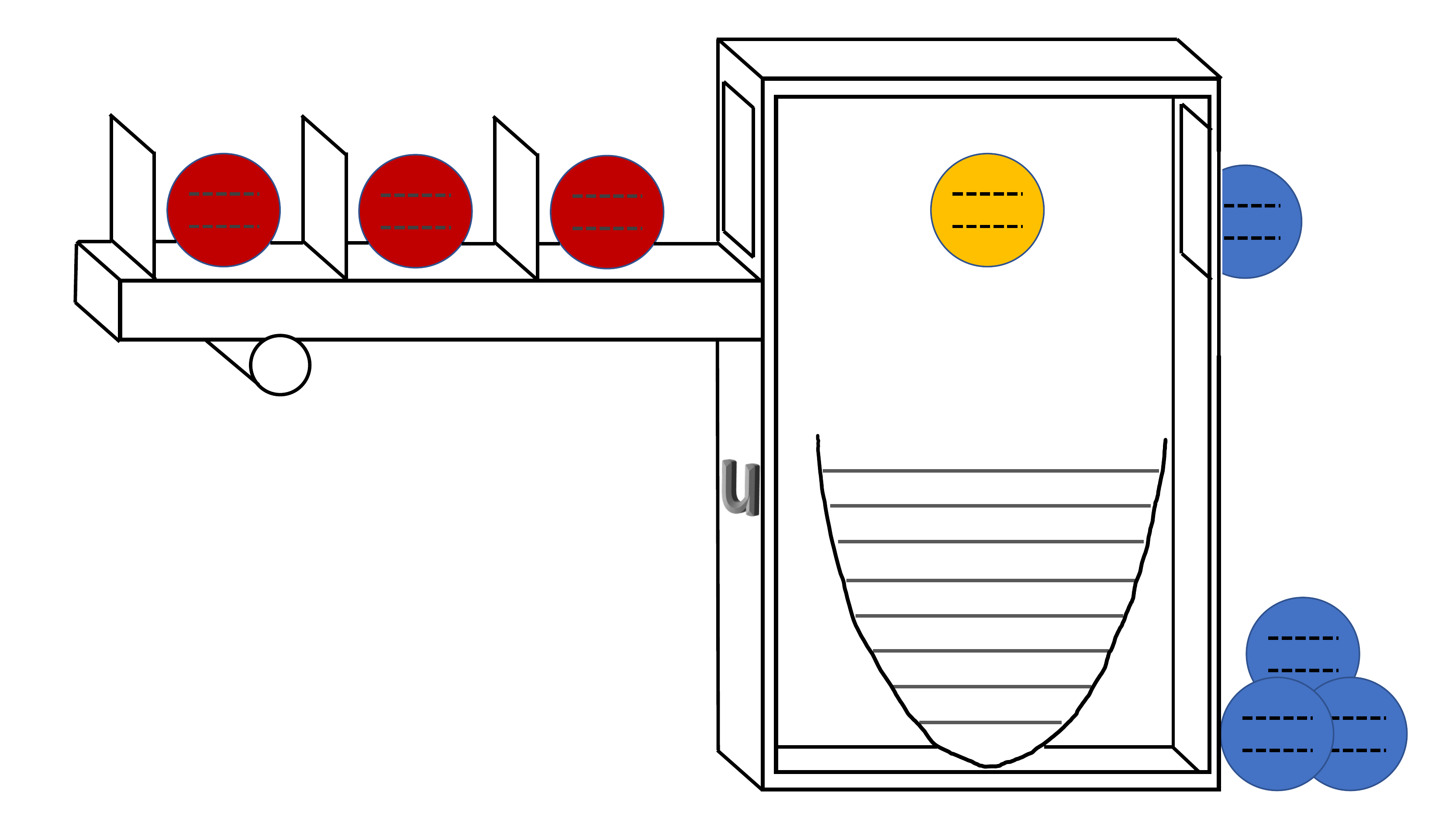}
\caption{\emph{A paradigm of a cyclic charging process and the Second Law.} Collection of identically prepared systems are one by one fed to a machine performing an energy-conserving unitary operation on a system and the battery (e.g., harmonic oscillator). The Second Law emerges in the asymptotic limit, after many elementary steps, which we postulate as: The ergotropy of the battery may increase indefinitely if and only if the systems are in a non-passive state. }
\label{charging_illustration}
\end{figure}

This paper aims to provide a handful of results that support the above formulation of the Second Law and to provide insight into the general features of the charging process (depending on the passivity/non-passivity of the charging systems). For this purpose, we adopt a collision-like model (\cite{Ciccarello2021, Salvia2022, Seah2021, Cattaneo2021, Rodrigues2019}, see \cite{Cusumano2022} for a recent review). As far as the conjecture looks simple, a closer inspection shows that it is entirely nontrivial, even in the simplest scenarios. First of all, the generalization to arbitrary quantum states with a coherent contribution to the ergotropy makes it, at the moment, quite intractable, even in the case of lowest dimensional systems. Thus, in this paper, we only concentrate on the diagonal case, which, nevertheless, reveals several interesting subtleties. Our main result is a matching of the diagonal case with the theory of Markov chains and identifying the passive states with the so-called {\it recurrent} Markov process and the non-passive case with {\it transient} one. We prove this correspondence for the two-level system and provide the numerical simulation for the higher dimensional ones. In this way, drawing work is portrayed as a transient process, where the battery is indeed charged with useful, non-passive energy. On the contrary, passive states cannot induce such transient processes, leading the battery instead to the passive stationary state. The mathematical notions adopted from the theory of Markov chains allow us also to distinguish between {\it strictly passive} and {\it maximally mixed} states. As we show for the two-level systems, the latter induces the so-called {\it null-recurrent} process, while strictly passive states correspond to {\it positive-recurrent} processes. This in particular reflects the fact that with a maximally mixed state, one can increase the energy of the battery indefinitely. Yet, the charging effect is asymptotically qualitatively different from the transient process, which only a non-passive state can induce.


\section{Charging process classification and the Second Law}
Let us first describe the general model of a charging process with explicit battery. We take $m$ identical copies of the quantum system $S$, i.e., $\xi^{\otimes m}$, that represent our charging resource (or fuel), and a single quantum state $\rho_0$, that represents the energy-storage device (i.e., a battery $B$). We want to charge the battery via the subsequent interactions of the type: $\rho_n \otimes \xi \to U \rho_n \otimes \xi U^\dag$, where $\rho_n = \Tr_S [U \rho_{n-1} \otimes \xi U^\dag]$ is the state of the battery after the $n$ elementary charging steps and $U$ is the unitary operator. In other words, we make use of copies of the state $\xi$ by sequentially coupling them to the current state of battery $\rho_n$. The system is later discarded, as illustrated in Fig. \ref{charging_illustration}. Within this general charging scenario, we further assume: 
\begin{enumerate}[label=\alph*)]
\item \emph{The battery is a harmonic oscillator system}. The energy spectrum of the oscillator is crucial for our understanding of the Second Law since, on the one hand, it is unbounded from above, and thus it describes an ideal battery with an infinite capacity able to capture the subtleties of the asymptotic limit, and on the other hand, it is bounded from below, i.e., it describes the physical system with the ground state (on the contrary to the ideal weight, see \cite{Lipka-Bartosik2021}). 
\item \emph{The unitary $U$ is strictly energy-preserving, i.e., $[U, H_S + H_B] = 0$, where $H_S$ and $H_B$ are free Hamiltonians of the system and the battery.} This assumption expresses the collisional type of the interaction, i.e., systems do not interact neither before nor after the coupling. From the charging perspective, it precludes any uncontrollable flow of the energy. 
\end{enumerate}

Now, let us consider the simplest charging process as possible. We take a two-level system (TLS) with the same frequency of the harmonic oscillator (for simplicity we put $\hbar \omega=1$), and with the diagonal state $\xi$ represented by the probabilities $s_1$ (for a ground state $\ket g $) and $s_2$ (for the excited state $\ket e$). Then, we take the charging unitary in the form:
\begin{eqnarray} \label{unitary_example}
    U =  \dyad{1,g} + \sum_{n=2}^\infty \big(\dyad{n,g}{n-1,e} + \dyad{n-1,e}{n,g} \big),
\end{eqnarray}
where $\ket{n,g(e)}$ is a product of the $n$-th eigenstate of the battery ($n=1,2,\dots$) and ground (excited) state of the TLS. Finally, we are interested in the final energy distribution of the battery after $m$ charging steps, given by $p_k(m) = \bra k \rho_m \ket k$. The particular realization of this process is presented in the Fig. \ref{fig:battery_distribution}.

\begin{figure}
    \includegraphics[width=0.45\textwidth]{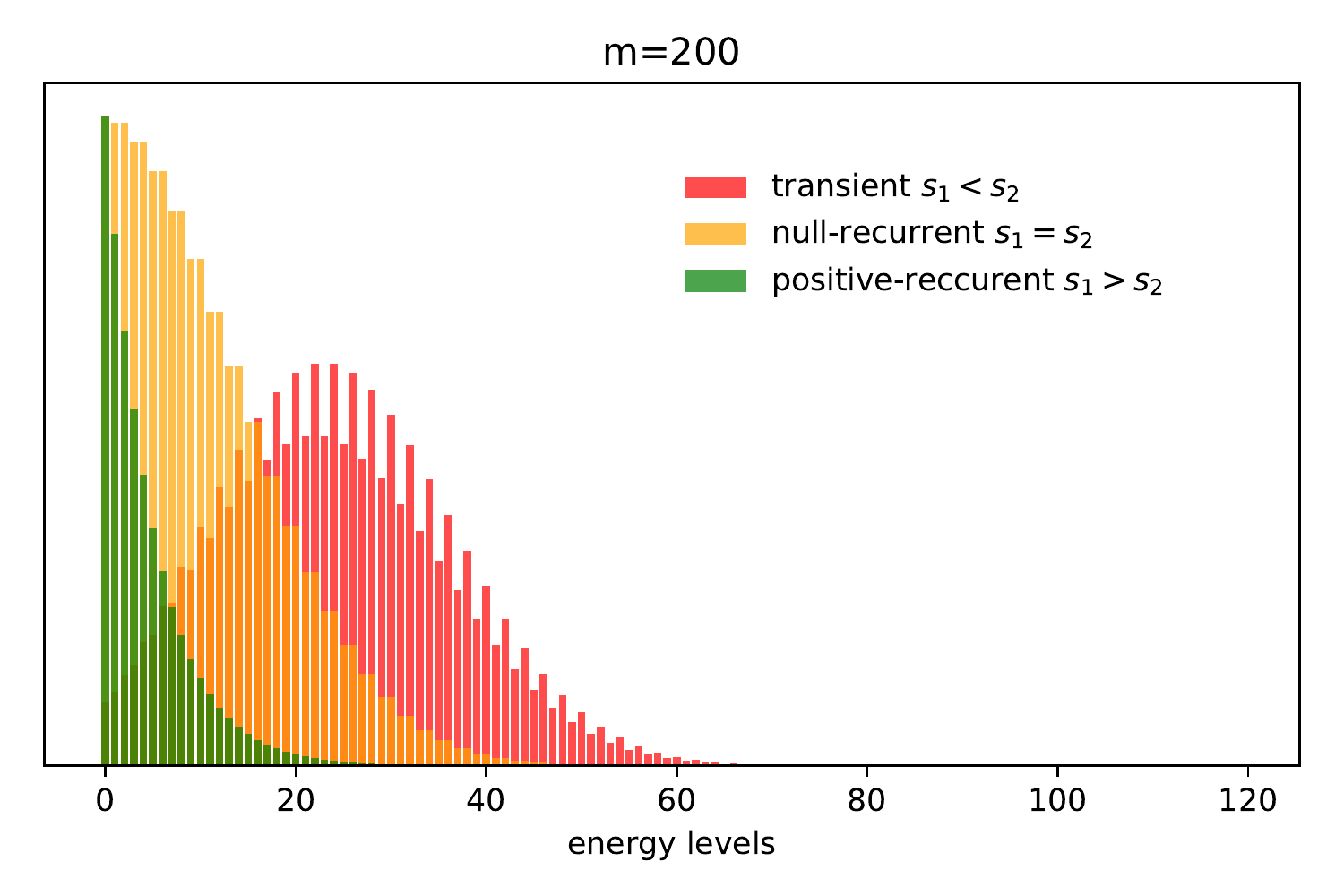}
    \includegraphics[width=0.45\textwidth]{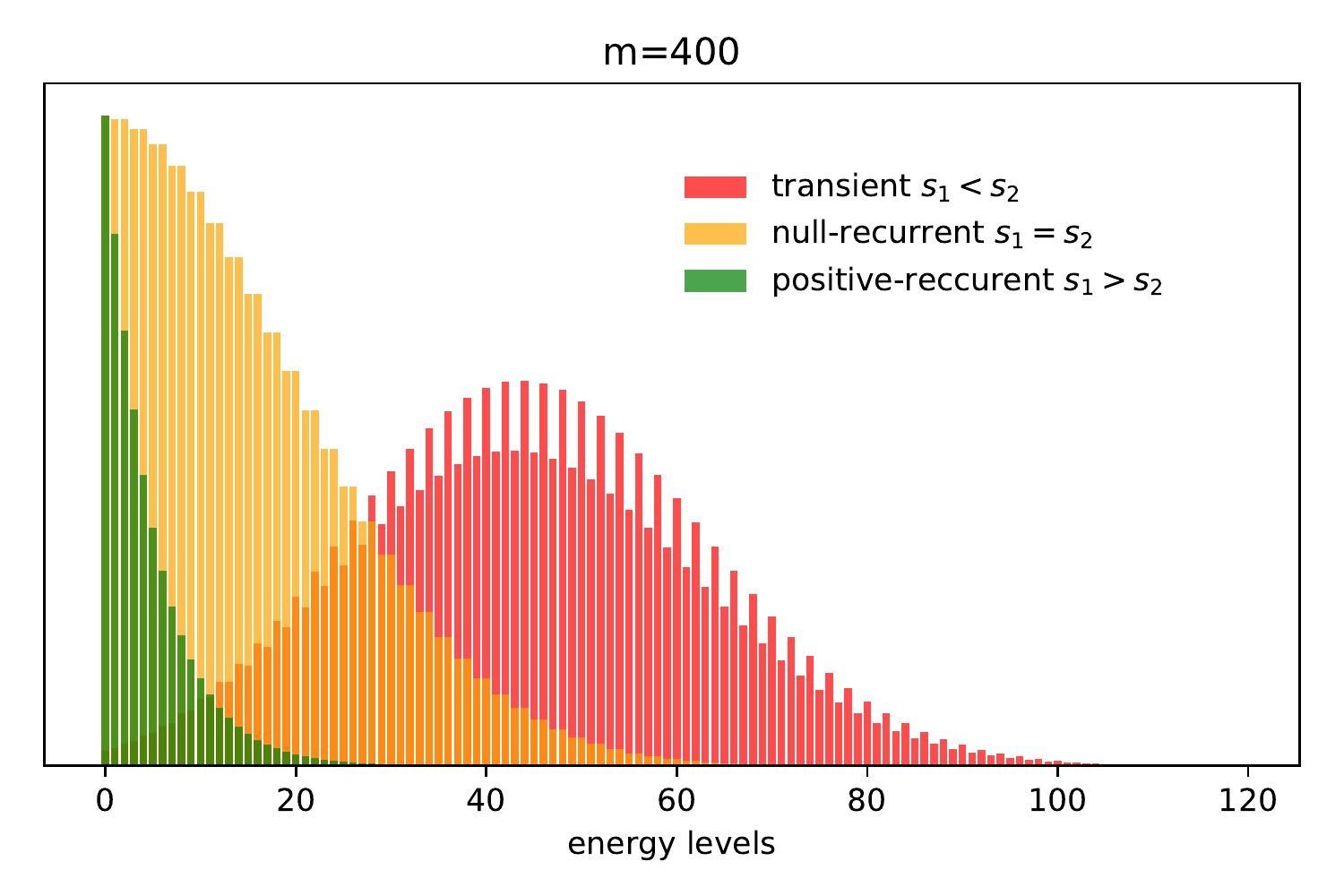}
    \includegraphics[width=0.45\textwidth]{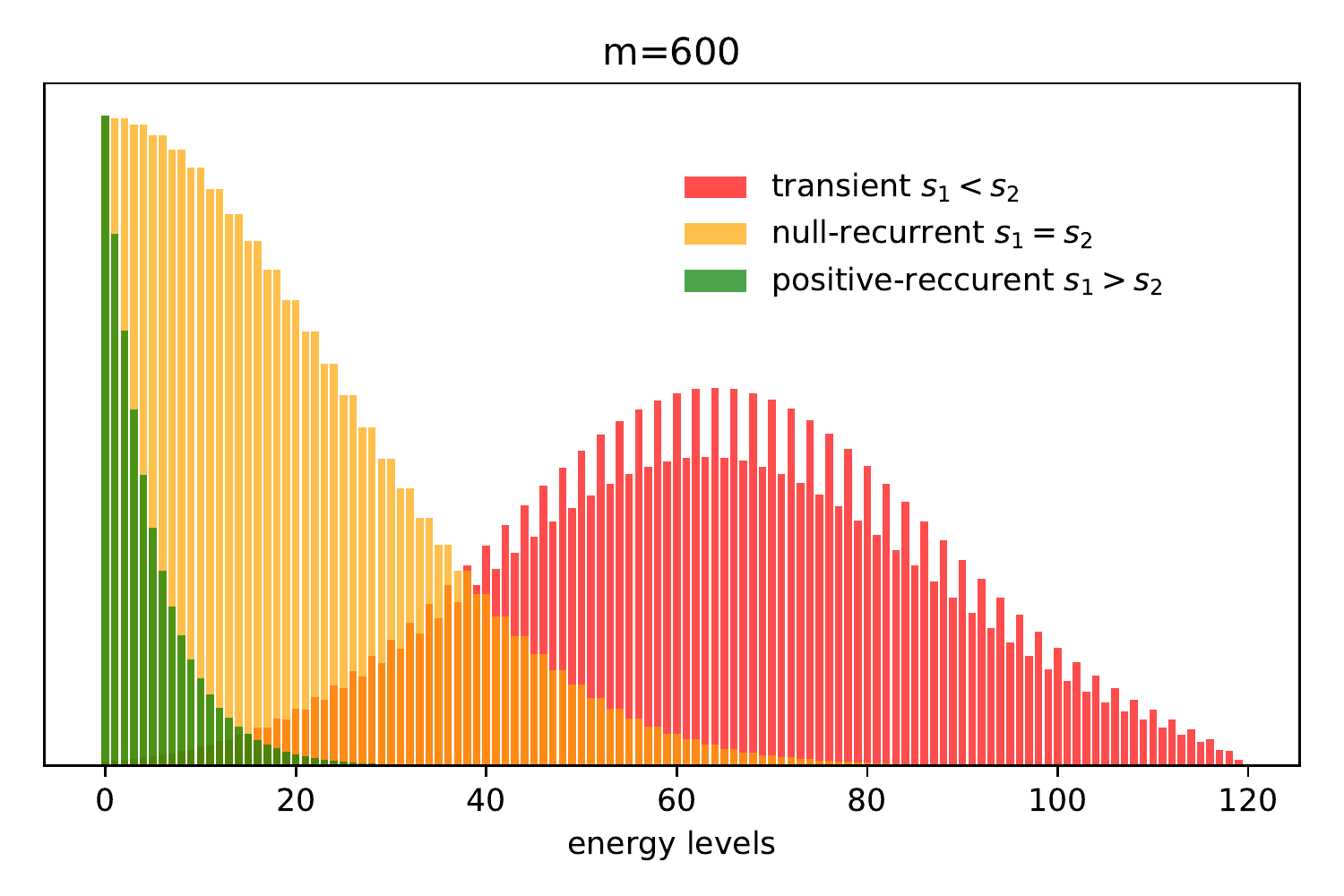}
\caption{\emph{The energy distribution of the battery after the $m$ charging steps starting from the ground state.} Three different charging process are present for the same unitary given by Eq. \eqref{unitary_example} but for different states of the charging TLSs. The green distribution, corresponding to the positive-recurrent process for the strictly-passive state of the TLS ($s_1 > s_2$), is the stationary process that quickly brings the battery to the Gibbs state. The orange distribution that represent the null-recurrent state (for maximally mixed state with $s_1=s_2$) is a non-stationary process that leads to the indefinite grow of the energy (but without accompanied increase of the ergotropy). The red distribution correspond to the non-stationary transient process (with the active state $s_1 < s_2$) that leads to the indefinite increase of the battery's ergotropy. Probability distributions are arbitrarily normalized. } 
\label{fig:battery_distribution}
\end{figure}

We observe here a very interesting behaviour: The charging process of the battery can be classified into a three distinct classes, which solely depends on the relation between probabilities $s_1$ and $s_2$. We classify those as:
\begin{enumerate}[label=(\roman*)]
\item ($s_1 < s_2$) \emph{Transient.} The state of the battery ``running away'' from the ground state, such that with growing number of charging steps ($m \to \infty$), the battery does not approach any stationary state and its both energy and (incoherent) ergotropy grows indefinitely. 
\item ($s_1 > s_2$) \emph{Positive-recurrent.} The state of the battery approaches a unique stationary state, i.e.,  the Gibbs state obeying the condition $p_k/p_{k+1} = s_1/s_2$ (when $m \to \infty$).  
\item ($s_1 = s_2$) \emph{Null-recurrent.} The state approaches the non-normalizable maximally mixed state, i.e., there is no stationary distribution and (when $m \to \infty$) the energy grows indefinitely (but not the ergotropy!). 
\end{enumerate}

We see that the asymptotic behavior of the battery is solely dictated by the passivity features of the fuel (state of the system $\xi$) and is robust to the initial state of the battery $\rho_0$. The presented paradigm expresses the Second Law of thermodynamics, showing that only the non-passive state provides an inexhaustible resource for the charging process. The crucial observation is a difference between \emph{transient} and \emph{null-recurrent} charging, where the former may increase the ergotropy of the battery, whereas the latter increases only the energy.      

Consequently, the fundamental question that should be addressed is whether this classification still holds for arbitrary energy-preserving unitary and for arbitrary states $\xi$ of the $d$-dimensional system \footnote{For non-diagonal states $\xi$ we generalize the notion of passivity.}? If the answer is positive, the introduced classification will be one of the most robust formulations of the Second Law for a charging process in quantum mechanics that is independent of any details of the process and initial condition of the battery. 
Thus, we postulate the Second Law in the form presented in the Introduction. In the following, we make a few first steps towards the general proof of this conjecture. 

\section{Process classification for diagonal states}
In this section, we introduce the methodology to prove the Second Law for diagonal states $\xi$. Basically, for diagonal states, the evolution of the battery energy distribution is given by the Markov chain, which, according to standard terminology, can be \emph{transient}, \emph{positive-recurrent} or \emph{null-recurrent}. This fits the classification based on particular observation made in the section before. Markov chain theory enables us to introduce the notions of \emph{transience} or \emph{positive/null-recurrence}, as described in the previous section, in a more rigorous way, at least for diagonal states $\xi$ (see Appendix \ref{markov_definitions_appendix}).

Let us then consider an elementary charging step $\rho \otimes \xi \to U \rho \otimes \xi U^\dag$ for $d$-dimensional diagonal state $\xi$ represented by the vector of probabilities $(s_1,s_2, \dots, s_d)$, and an initial distribution of the battery given by probabilities $p_k$ (for $k \in \mathbb{N})$. Strict energy-preservation of the unitary enforces the block-diagonal structure $U=\bigoplus_{i=1}^{\infty} u^{(i)}$ (see Appendix \ref{charging_process_appendix}).
Then, the new energy distribution (after the charging step) is given by:
\begin{eqnarray} \label{Markov_process}
   p_k' = \sum_{m=1}^\infty T_{km} p_m
\end{eqnarray}
where 
\begin{eqnarray} \label{transition_matrix}
   T_{km} = \sum_{n=1}^\infty \sum_{i,j = 1}^{\min(n,d)}  |u^{(n)}_{ij}|^2 s_j \ \delta_{n-j+1,m} \ \delta_{n-i+1,k},
\end{eqnarray}
and $u_{ij}^{(n)}$ are the elements of an arbitrary unitary (in particular, it implies $\sum_i |u_{ij}^{(n)}|^2 = \sum_j |u_{ij}^{(n)}|^2 = 1$). For the incoherent state $\xi$, the diagonal part of the battery transforms independently of the off-diagonal elements (coherences). In accordance, we can identify the diagonal evolution with the stochastic process, namely the Markov chain defined on the space of energy levels (natural numbers). From now on, the main object of our interest is the transition matrix $T$ (Eq. \eqref{transition_matrix}) of that process. Elements $T{i,j}$ of the matrix set probabilities for transitions between energy levels of the battery, as depicted in Fig. \ref{TFig}.
\begin{figure}
    \includegraphics[width=0.5\textwidth]{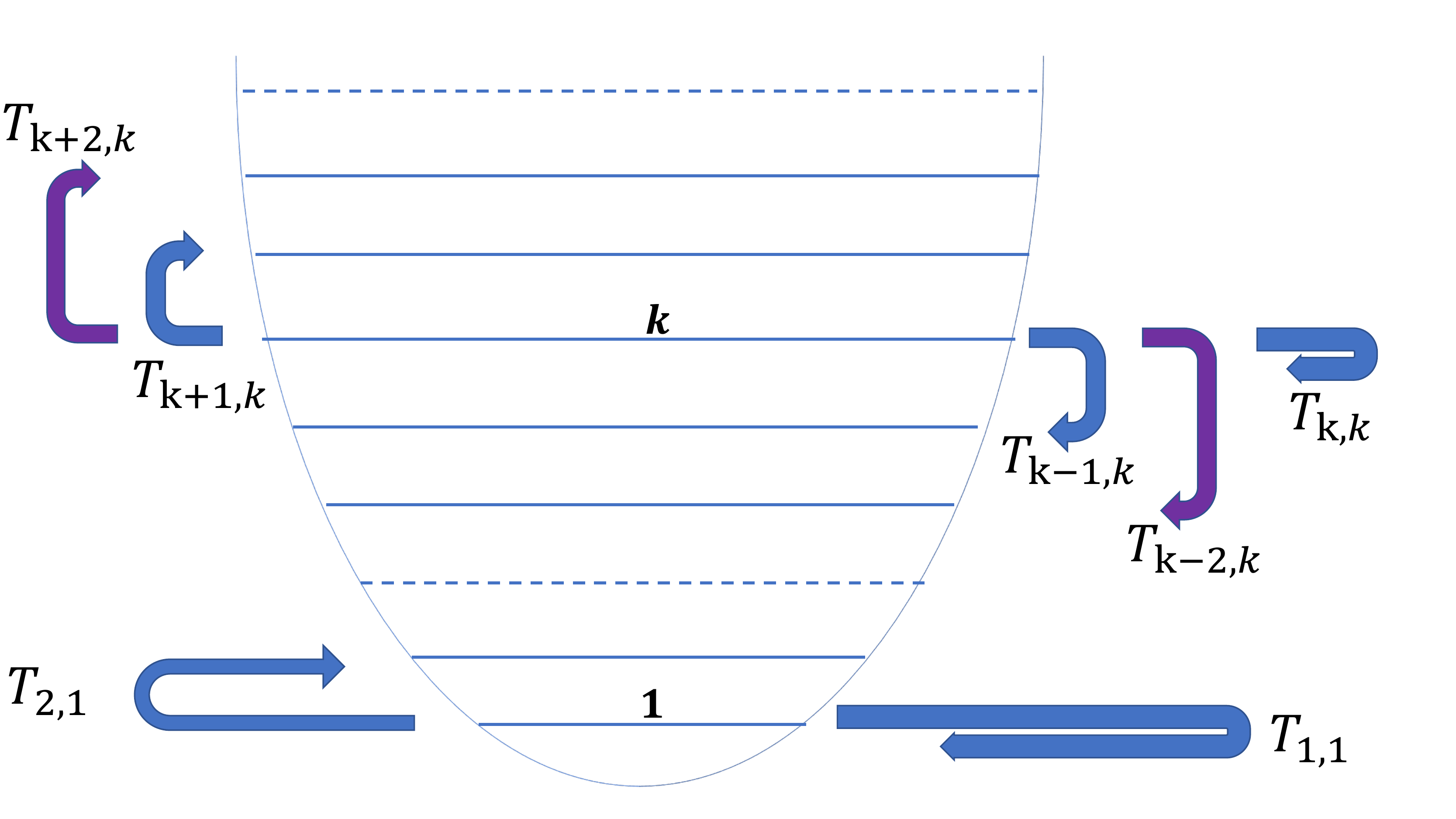}
\caption{Action of elements of the transition matrix $T$ on populations of the harmonic oscillator battery. For the interaction with TLS in the collisional model, only nearest neighbour transitions are allowed (blue arrows) within one collision. Possible transitions are restricted from the bottom by the existence of the ground state.\label{TFig}}
\end{figure}
Then, according to the theory of Markov chains we define:

\begin{definition}[\emph{Transient, positive- and null-recurrent state}] \label{transient_and_recurrent_definition}
We define the \emph{expected return time} for the Markov chain with transition matrix $T$ as 
\begin{eqnarray}
   \langle \tau_k \rangle = \sum_{n=1}^\infty n \ T^n_{kk}.
\end{eqnarray}
Then, the $k$-th state of the chain is
\begin{align}
    \text{transient:} \quad &\forall_n \ T^n_{kk} < 1; \\
    \text{positive-recurrent:} \quad  &\exists_n  \ T^n_{kk} = 1 \quad \text{and} \quad \langle \tau_k \rangle < \infty; \\
    \text{null-recurrent:} \quad &\exists_n  \ T^n_{kk} = 1 \quad \text{and} \quad \langle \tau_k \rangle = \infty.
\end{align}   
\end{definition}
Let us first point out the most essential conclusions coming from this classification:
\begin{enumerate}
    \item If the Markov chain is irreducible, then all of its states are of the same type (i.e., recurrent or transient). Thus, we characterize the irreducible chains as recurrent or transient.
    \item The positive-recurrent chain has a unique stationary distribution, such that:
   $p_k = \sum_{m=1}^\infty T_{km} p_m.$
    \item There is no stationary distribution for transient or null-recurrent chains.
    \item For finite dimensional spaces, all states are positive-recurrent.
\end{enumerate}
In the following, we will prove that the processes plotted in Fig. \ref{fig:battery_distribution} are indeed transient and positive/null-recurrent with respect to the above classification. 
In particular we see from the Markov chain theory that the transient and null-recurrent processes are non-stationary, whereas the positive-recurrent is the only one stationary. From the physical point of view, we imagine that the repeated charging via the positive-recurrent process drags the state of the battery to the stationary solution and the charging has to saturate at some point (compare with the green distribution in Fig. \ref{fig:battery_distribution}). On the other hand, the transient and null-recurrent are the processes that describes the cyclic charging, but as we see from Definition \ref{transient_and_recurrent_definition} and based on the previous observations, they are very distinct in their asymptotic properties. We believe that only transient chain represents the proper charging of the battery which is accompanied with the ``indefinite increase of the ergotropy", and show it below for the TLS case.

One technical assumption has to be mentioned. In the following, we assume that the charging process is irreducible, which means that all energy states are connected, such that after some particular number of steps, there is a non-zero probability of jumping between two arbitrary states (see Appendix \ref{markov_definitions_appendix} for more details). The reason behind this is to capture the subtleties connected with the potentially infinite charging process. Otherwise, an infinite space of the harmonic oscillator is reduced to a finite-dimensional, and the process always tends to stationary distribution (as we see from remark 4).

\subsection{Two-level systems}
For the charging systems given by two-level systems, we propose the following theorem:
\begin{theorem} \label{theorem_recurrent}
The irreducible Markov chain with the transition matrix given by Eq. \eqref{transition_matrix} for the two-dimensional system with probabilities $(s_1, s_2)$ is positive-recurrent if $s_1 > s_2$ and null-recurrent if $s_1 = s_2$. 
\end{theorem}
We refer the reader to Appendix \ref{two_level_system_appendix} for the proof, which is based on the Foster's theorem \cite{Bremaud1999}.

The Theorem \ref{theorem_recurrent}, for a diagonal TLS, simply states: $passive \ state \implies recurrent \ process$. As a corollary, we immediately get the following negation: $transient \ process \implies active \ state$. However, we stress that in general $active \ state \centernot \implies transient \ process$ since the active state provides only the possibility of a transient process, which has to be accompanied by a selection of a proper unitary. In particular, for a unitary \eqref{unitary_example} and an active state ($s_1 < s_2$), the associated process is transient. 
\begin{proposition}
For the positive-recurrent process (with $s_1 > s_2$), the stationary distribution is given by the Gibbs state, satisfying (for $k \ge 1$):
\begin{eqnarray}
   p_{k+1} = \frac{s_2}{s_1} \ p_{k},
\end{eqnarray}
i.e., the battery approaches the equilibrium distribution with the same temperature as the collection of charging TLSs. 
\end{proposition}
One important remark must be pointed out, namely, that the stationary state for the positive-recurrent process does not depend on the applied form of a unitary protocol but solely on the TLS state. In particular, it means that for a strictly-passive TLS, one can change the unitary between the charging steps, yet the battery still approaches the unique stationary state given by the Gibbs distribution. As we will see in the following subsection, it is not valid for multi-level systems, for which the stationary state depends on the unitary process.      

\subsection{Multi-level systems}
For multi-level systems (i.e., with dimension $d>2$), the situation complicates drastically, so we could not provide any general proof for the charging classification. Nevertheless, we provide a numerical simulation to support our conjectures for higher dimensional systems. In particular, we sample over a random transition matrices of the form:
\begin{eqnarray} \label{quasi_transition_matrix}
   \tilde T_{km} = \sum_{n=1}^\infty \sum_{i,j = 1}^{\min(n,d)}  B^{(n)}_{ij} s_j \ \delta_{n-j+1,m} \ \delta_{n-i+1,k},
\end{eqnarray}
where $B^{(n)}_{ij}$ are the elements of random bistochastic matrices. We stress that the set of all matrices $T$, represented by the elements given by Eq. \eqref{transition_matrix}, forms a subset of all transition matrices $\tilde T$ since the former is constructed out of the unistochastic matrices (with elements $|u_{ij}^{(n)}|^2$) instead of more general bistochastic ones. Nevertheless, since we enlarge the class of the charging processes (for better efficiency of the numerical simulation), our general conclusions are still valid for those originating from the unitary evolution.

We want to highlight the most crucial observations. Primarily, we have simulated the change of the battery's ergotropy for the repeated charging process with the random transition matrix \eqref{quasi_transition_matrix} and the random state of a five-dimensional charging system. Among all the sampled realizations, we have never observed the trajectory that contradicts the Second Law implication:  $indefinite \  increase \ of \ the \ ergotropy  \implies active \ state$. As an illustration, in the left panel of Fig. \ref{fig:stationary_states}, we present a sample of twenty trajectories that support this claim. 

The second general observation is the following: \emph{The battery charged via strictly-passive states leads to the passive stationary state (but not necessarily the Gibbs state).} It is observed from the discussed before random trajectories where it is seen that the ergotropy tends to zero for passive states. However, contrary to the two-level systems, where the stationary state was uniquely defined as a Gibbs state (since any passive TLS is necessarily also a Gibbs state), for higher dimensions, it is given by a more general passive state. Moreover, as we present in the right panel of Fig. \ref{fig:stationary_states}, this stationary passive state also depends on the form of the applied transition matrix.    


\begin{figure}
    \includegraphics[width=0.47\textwidth]{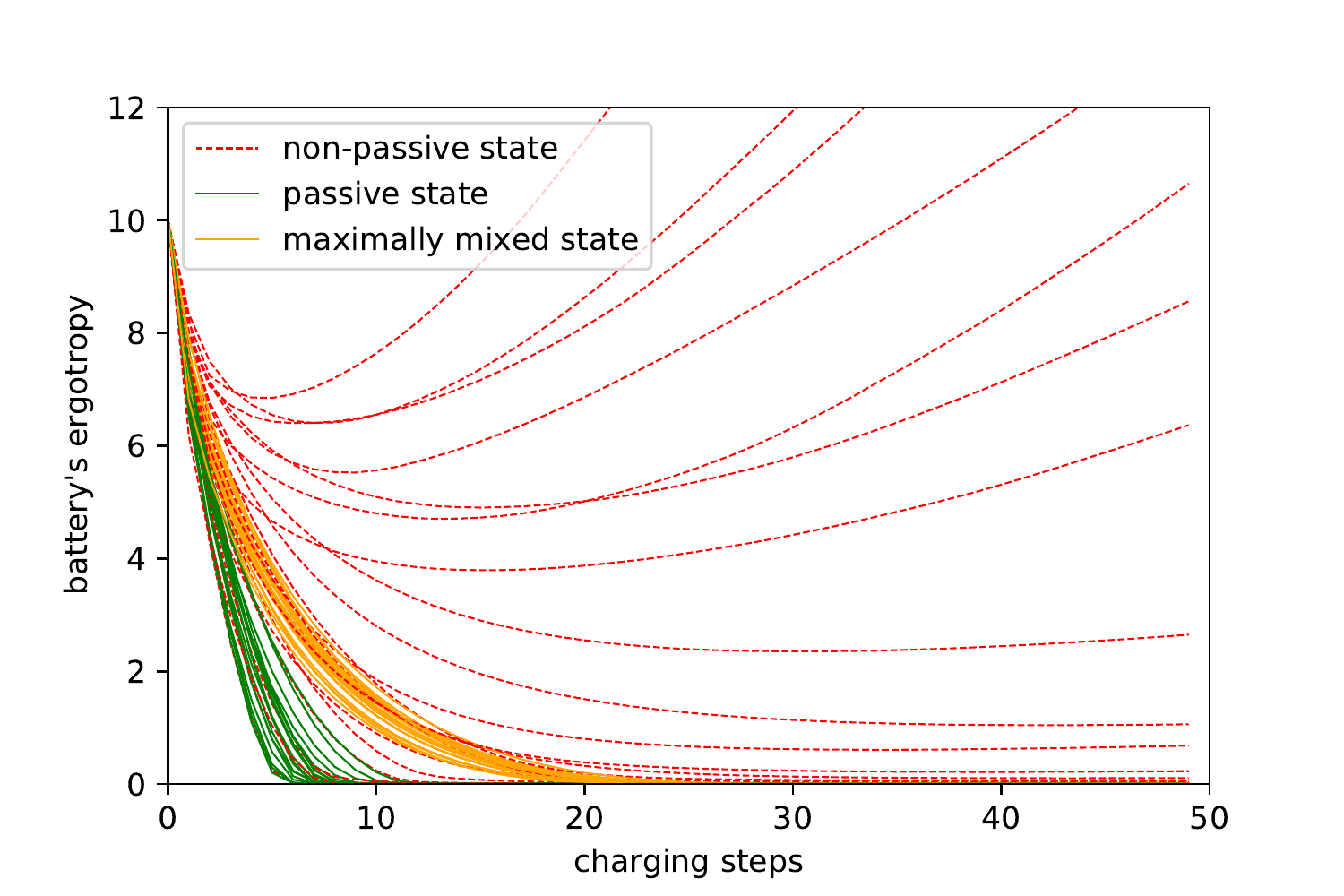}
    \includegraphics[width=0.47\textwidth]{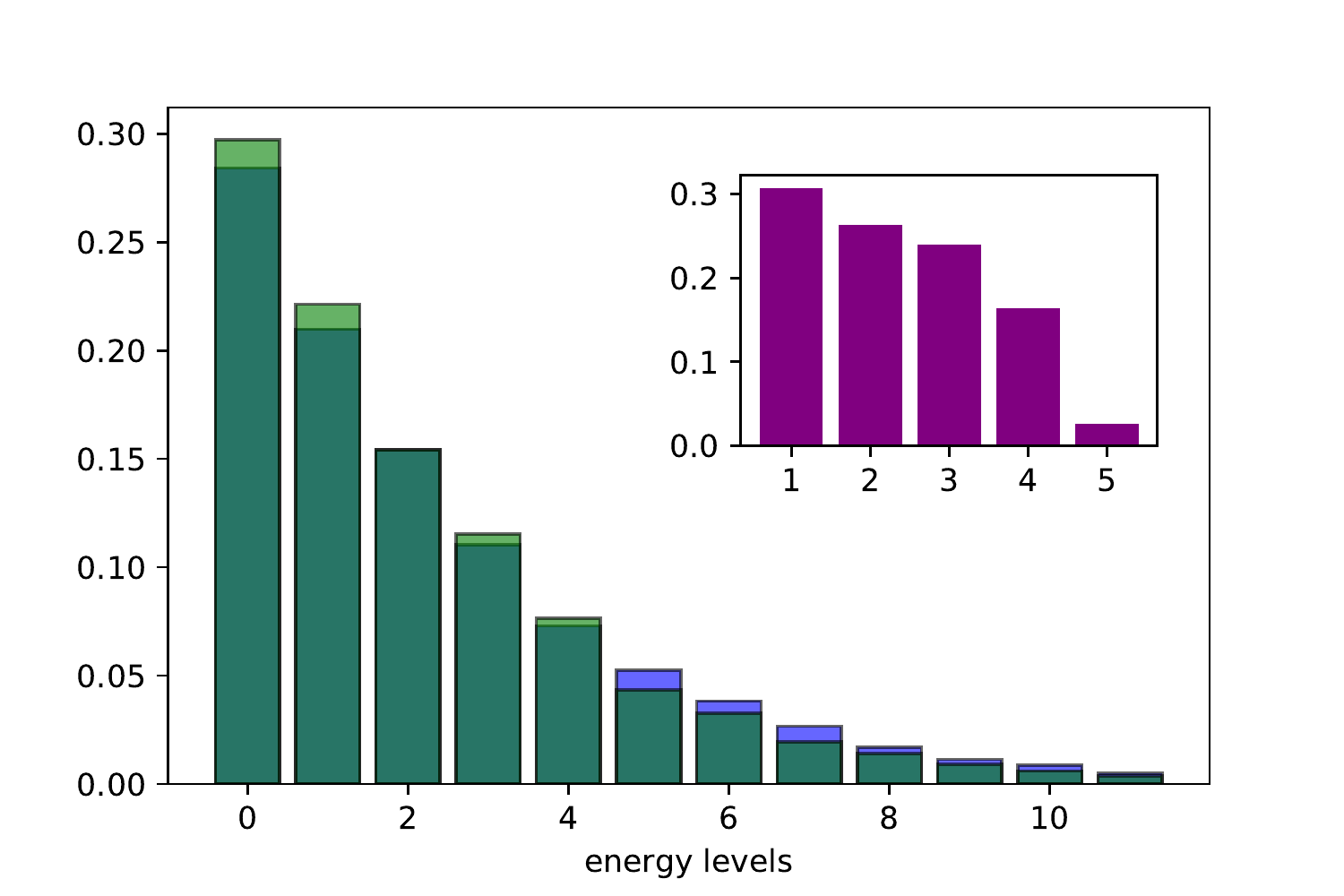}
\caption{\emph{Battery charging through collisions with five-level systems.} \textbf{Left}: Trajectories of battery's ergotropy for random transition matrices (\ref{quasi_transition_matrix}), and random charging states (grouped as: passive -- green, non-passive -- red, and maximally mixed -- orange). Initial state of the battery is fixed and pure. Consecutive charging with passive and maximally mixed states results in erasure of battery ergotropy, while some trajectories for non-passive states show either saturation or indefinite increase of ergotropy. \textbf{Right}: Histogram of two distinct stationary distributions of battery's population, resulting from application of two different transition matrices (\ref{quasi_transition_matrix}). Light green and blue rectangles mark excesses of respective distributions. Inset shows population distribution of the common passive state used for the charging.  } \label{fig:stationary_states}
\end{figure}

\section{Outlook and Conclusions}
A closer analysis of the fully quantum charging process for a generic battery (i.e., without any additional assumptions, like translational symmetry) shows that the Second Law cannot be formulated for a single cycle but only emerges in the asymptotic limit. This reveals the true meaning of the ``cyclicity'' in Planck's formulation, which involves not only the working body but also the work reservoir. In accordance, we postulate the Second Law as inability of the indefinite increase of the battery's ergotropy via the repeated charging by passive states. 

In this paper, we open a broad program to prove the introduced Second Law and characterize the charging process in general. The introduced methods related to the theory of Markov chains suggest that the general characterization of the charging can be done without any details of the applied operation but is solely based on the passivity/non-passivity of the charging fuel. Nevertheless, our insight into the problem is just the tip of the iceberg. Let us briefly point out the open questions within this paradigm. 

Firstly, we only touch on the problem for diagonal states, which very nicely corresponds to the theory of Markov chains. However, even for this semi-classical problem, the general classification of the processes has been done only for the lowest dimensional system (whereas for multi-level systems, we provide numerical observations). Moreover, as far as the classification of transient and positive/null-recurrent Markov chains expresses the intuitions in  mathematical terms, 
we do not provide a general proof that the transient process leads to an indefinite increase of the ergotropy, whereas the null-recurrent does not. Moreover, for higher dimensional systems, the situation even for positive-recurrent processes is much more complicated since the stationary distribution now depends on the form of the unitary operator. The crucial observation is that battery in general does not approach the Gibbs state, but rather a more general - passive state. However, this also requires a general proof. 

Secondly, the whole framework should be generalized for non-diagonal quantum states. In this case, we need to include the coherent contribution to the battery's ergotropy. Consequently, we cannot restrict the analysis only to the diagonal part of the battery density matrix. Moreover, while the introduced concepts (i.e., transience and null/positive-recurrence) have a  well defined meaning for the diagonal case,  we cannot 
represent evolution as a stochastic process for a general quantum scenario; therefore, we cannot directly use the definitions from the Markov's theory. Thus, a very interesting mathematical question arises: What is the meaning of the transience/recurrence in a fully quantum scenario? A hint on how to approach it may be delivered by Foster's theorem \cite{Bremaud1999}, which provides equivalent conditions for transience/recurrence that can still be generalized for generators acting on the Hilbert space. This, however, needs further careful investigations.

\subsection*{Acknowledgements}
The authors thank Konrad Ja{\l}owiecki for providing us with a code to simulate the random transition matrices. We acknowledge support from the Foundation for Polish Science through IRAP project co-financed by EU within the Smart Growth Operational Programme (contract no.2018/MAB/5).

\bibliography{bib}

\newpage
\appendix

\section{Subsequent charging process} \label{charging_process_appendix}
We consider a $d$-dimensional qudit (system $\mathcal{S}$) with the Hamiltonian $H_S = \sum_{k=1}^{d} \omega \dyad{k}_S$ and the battery $\mathcal{B}$ given by the harmonic oscillator with the Hamiltonian $H_B = \sum_{k=1}^\infty k \omega \dyad{k}_B$. The charging process of the battery after $n$ cycles is given by: 
\begin{equation}
    \rho_n = \Tr_S[U_n \ \rho_{n-1} \otimes \xi \ U_n^\dag]    
\end{equation}
where $U_n$ is the energy-conserving unitary operator, such that it satisfies $[U_n, H_S + H_B] = 0$ and $\xi$ is some fixed state of the system $\mathcal{S}$.

Any energy-conserving unitary can be represented in the block-diagonal form, i.e., 
\begin{eqnarray} \label{energy_conserving_unitary}
   U =  \sum_{n=1}^\infty \sum_{i,j = 1}^{\min(n,d)} u^{(n)}_{ij} \dyad{n+1-i}{n+1-j}_B \otimes \dyad{i}{j}_S
\end{eqnarray}
such that it is given by the direct sum $U = u^{(1)} \oplus u^{(2)} \oplus \dots$ of the unitary blocks $u^{(k)}$. The matrix representation of the unitary $U$ is given by:  
\begin{eqnarray}
   U = \begin{pmatrix}
   u_{11}^{(1)} &  &  \\
    &
    \begin{matrix}
        u_{11}^{(2)} & u_{12}^{(2)} \\
        u_{21}^{(2)} & u_{22}^{(2)}
    \end{matrix} 
    &  \\
     &  &
    \begin{matrix}
        u_{11}^{(3)} & u_{12}^{(3)} & u_{13}^{(3)} \\
        u_{21}^{(3)} & u_{22}^{(3)} & u_{23}^{(3)} \\
        u_{31}^{(3)} & u_{32}^{(3)} & u_{33}^{(3)} \\
    \end{matrix} \\
    & & & \ddots \\
    & & & & 
    \begin{matrix}
        u_{11}^{(d)} & \dots & u_{1 d}^{(d)} \\
        \vdots & \ddots & \vdots \\
        u_{d 1}^{(d)} & \dots & u_{d d}^{(d)} \\
    \end{matrix} \\
    & & & & & 
    \begin{matrix}
        u_{11}^{(d+1)} & \dots & u_{1 d}^{(d+1)} \\
        \vdots & \ddots & \vdots \\
        u_{d 1}^{(d+1)} & \dots & u_{d d}^{(d+1)} \\
    \end{matrix} \\
    & & & & & & \ddots
   \end{pmatrix}.
\end{eqnarray}
One should notice that the dimension of the first block is one (which together with unitary condition gives simply $u_{11}^{(1)} = 1$), while for other blocks it grows up to the dimension $d$ of the system.    

\subsection{Markov chain for a charging process via diagonal states}
Let us restrict our analysis only to diagonal states of the system, such that:
\begin{eqnarray}
   \xi = \sum_{k=1}^d s_k \dyad{k}_S,
\end{eqnarray}
where $(s_1, s_2, \dots, s_d)$ forms an arbitrary vector of probabilities. In accordance, we introduce the following classification of the states.

\begin{definition} [\emph{Passive state, active state and maximally mixed state}]
The state of the system is \emph{passive} if $s_k \le s_{k+1}$ for each $k = 1, \dots, d-1$. The passive system is additionally \emph{maximally mixed} if all probabilities are equal. Finally, the system is \emph{active} if it is not passive.
\end{definition}

Next, let us consider a single step of the charging process via the diagonal state $\xi$, i.e., we calculate the state of the battery after a single charging:
\begin{equation} \label{charging_process}
    \rho' = \Tr_S[U \rho \otimes \xi U^\dag],    
\end{equation}
where $\rho = \sum_{i,j} \rho_{ij} \dyad{i}{j}_B$ is some arbitrary initial state. In the following, for simplicity we put: $\dyad{i}{j}_B \otimes \dyad{m}{n}_S \equiv \dyad{i,m}{j,n}$. The unitary evolution gives us:
\begin{equation}
\begin{split}
   & U \rho \otimes \xi U^\dag = 
   \begin{gathered}
    \left(\sum_{n=1}^\infty \sum_{i,j = 1}^{\min(n,d)} u^{(n)}_{ij} \dyad{n+1-i,i}{n+1-j,j}\right)     \left(\sum_{k,l=1}^\infty \sum_{m=1}^d \rho_{kl} s_m \dyad{k,m}{l,m}  \right)  \\ 
    \left(\sum_{n'=1}^\infty \sum_{i',j' = 1}^{\min(n',d)} u^{(n')}_{i'j'} \dyad{n'+1-i',i'}{n'+1-j',j'} \right)^\dag 
   \end{gathered} \\
   &= 
   \begin{gathered}
       \left(\sum_{n=1}^\infty \sum_{i,j = 1}^{\min(n,d)} u^{(n)}_{ij} \dyad{n+1-i,i}{n+1-j,j}\right) \left(\sum_{k,l=1}^\infty \sum_{m=1}^d \rho_{kl} s_m \dyad{k,m}{l,m}  \right)  \\
       \left(\sum_{n'=1}^\infty \sum_{i',j' = 1}^{\min(n',d)} (u^{(n')}_{i'j'})^* \dyad{n'+1-j',j'}{n'+1-i',i'} \right) 
   \end{gathered} \\
&= 
\begin{gathered}
\sum_{n,n',k,l=1}^\infty \sum_{m=1}^d \sum_{i,j = 1}^{\min(n,d)} \sum_{i',j' = 1}^{\min(n',d)} u^{(n)}_{ij} \rho_{kl} s_m (u^{(n')}_{i'j'})^* 
\braket{n+1-j,j}{k,m} \braket{l,m}{n'+1-j',j'} \dyad{n+1-i,i}{n'+1-i',i'}
\end{gathered} 
\end{split}
\end{equation}
And after tracing out the system, we get:
\begin{align}
   &\Tr_S[U \rho \otimes \xi U^\dag] = \\ &=\sum_{n,n',k,l=1}^\infty \sum_{m=1}^d \sum_{i,j = 1}^{\min(n,d)} \sum_{i',j' = 1}^{\min(n',d)} u^{(n)}_{ij} \rho_{kl} s_m (u^{(n')}_{i'j'})^* \delta_{n+1-j,k} \delta_{j,m} \delta_{n'+1-j',l} \delta_{m,j'} \delta_{i,i'} \dyad{n+1-i}{n'+1-i'} \\
   &= \sum_{n,n',k,l=1}^\infty \sum_{i,j = 1}^{\min(n,d)} \sum_{i',j' = 1}^{\min(n',d)} u^{(n)}_{ij} \rho_{kl} s_j (u^{(n')}_{i'j'})^* \delta_{n+1-j,k} \delta_{n'+1-j',l} \delta_{j,j'} \delta_{i,i'} \dyad{n+1-i}{n'+1-i'} \\
   &= \sum_{n,n'=1}^\infty \sum_{i,j = 1}^{\min(n,d)} \sum_{i',j' = 1}^{\min(n',d)} u^{(n)}_{ij} \rho_{n-j+1,n'-j'+1} s_j (u^{(n')}_{i'j'})^* \delta_{j,j'} \delta_{i,i'} \dyad{n+1-i}{n'+1-i'}
\end{align}
Finally, we are interested in new set of probabilities (after the charging process) of occupying the energy states, i.e., 
\begin{align}
   p_k' = \bra{k} \Tr_S[U \rho \otimes \xi U^\dag] \ket{k} 
   &= \sum_{n,n'=1}^\infty \sum_{i,j = 1}^{\min(n,d)} \sum_{i',j' = 1}^{\min(n',d)} u^{(n)}_{ij} \rho_{n-j+1,n'-j'+1} s_j (u^{(n')}_{i'j'})^* \delta_{j,j'} \delta_{i,i'} \delta_{n-i+1,k} \delta_{n'-i'+1,k} \\
   &= \sum_{n=1}^\infty \sum_{i,j = 1}^{\min(n,d)}  |u^{(n)}_{ij}|^2 s_j p_{n-j+1}  \delta_{n-i+1,k} = \sum_{m=1}^\infty \sum_{n=1}^\infty \sum_{i,j = 1}^{\min(n,d)}  |u^{(n)}_{ij}|^2 s_j \delta_{n-j+1,m} \delta_{n-i+1,k} p_m \\
   & = \sum_{m=1}^\infty \left(\sum_{n=1}^\infty \sum_{i,j = 1}^{\min(n,d)}  |u^{(n)}_{ij}|^2 s_j \delta_{n-j+1,m} \delta_{n-i+1,k} \right) p_m  
\end{align}
The last formula describes the transformation rule of the battery diagonal part, which, importantly, transforms independently from the off-diagonal part (coherences). Thus, the expression in the bracket can be identified as the transition matrix of the stochastic process: 
\begin{eqnarray}
   T_{km} = \sum_{n=1}^\infty \sum_{i,j = 1}^{\min(n,d)}  |u^{(n)}_{ij}|^2 s_j \ \delta_{n-j+1,m} \ \delta_{n-i+1,k},
\end{eqnarray}
such that
\begin{eqnarray}
   p_k' = \sum_{m=1}^\infty T_{km} p_m.
\end{eqnarray}
\begin{proposition}[Markov chain of the charging process]
The charging process given by Eq. \eqref{charging_process} (for a diagonal $\xi$) defines a Markov chain on the space of natural numbers (energy eigenstates), such that the probability of moving from the state $m$ to the state $k$ via the single charging step is equal to $T_{km}$.
\end{proposition}

Notice that since the elements $u_{ij}^{(n)}$ represents the unitary, the elements $|u_{ij}^{(n)}|^2$ forms the bistochastic matrix, such that $\sum_i |u_{ij}^{(n)}|^2 = \sum_j |u_{ij}^{(n)}|^2 = 1$.

\subsection{Properties of the Markov chains} \label{markov_definitions_appendix}
In this section we consider the Markov chain introduced before, defined on the space of natural numbers, and fully characterize by the transition matrix $T$. 

\begin{definition}[Irreducibility]
A Markov chain with transition matrix $T$ is \textit{irreducible} if for any $k,m \in \mathbb{N}$ there exist a natural number $n$ such that 
\begin{eqnarray}
   (T^n)_{km} > 0.
\end{eqnarray}
\end{definition}
This property means that all states are connected with each other, i.e., after repeating the charging process, one can reach any state (starting from arbitrary initial one). 


\begin{definition}[\emph{Transient, positive- and null-recurrent state}]
Let us define the \emph{expected return time} for the Markov chain with transition matrix $T$ as 
\begin{eqnarray}
   \langle \tau_k \rangle = \sum_{n=1}^\infty n \ T^n_{kk}.
\end{eqnarray}
Then, the $k$-th state of the chain is
\begin{align}
    \text{transient:} \quad &\forall_n \ T^n_{kk} < 1; \\
    \text{positive-recurrent:} \quad  &\exists_n  \ T^n_{kk} = 1 \quad \text{and} \quad \langle \tau_k \rangle < \infty; \\
    \text{null-recurrent:} \quad &\exists_n  \ T^n_{kk} = 1 \quad \text{and} \quad \langle \tau_k \rangle = \infty.
\end{align}   
\end{definition}

\begin{remark}
If the Markov chain is irreducible all of its states are of the same type (i.e., recurrent or transient). Thus, we characterize the irreducible chains as recurrent or transient.
\end{remark}
\begin{remark}
The positive-recurrent chain has a unique stationary distribution, such that
\begin{eqnarray}
   p_k = \sum_{m=1}^\infty T_{km} p_m.
\end{eqnarray}
There is no stationary distribution for transient or null-recurrent chains.
\end{remark}

\begin{proposition}[Recurrence and transience criterion] \label{reccurent_transient_conditions}
Let us consider an irreducible Markov chain with transition matrix $T$ and a function $f_k > 0$ (for all $k \in \mathbb{N})$, such that
\begin{eqnarray} \label{function_inequality}
  \displaystyle\mathop{\forall}_{m \in \mathbb{N} \setminus A} \sum_k T_{km} f_k \le f_m,
\end{eqnarray}
where $A \subset \mathbb{N}$ is a finite non-empty set. Then, if there exist $f_k$ such that:  
\begin{enumerate}

    \item $f_k \to \infty$ as $k \to \infty$, the chain is recurrent;
    \item $f_k < \inf_{m \in A} f_m$ for at least one $k \in \mathbb{N} \setminus A$, the chain is transient. 
\end{enumerate}  
\end{proposition}

\section{Two-level system} \label{two_level_system_appendix}
We consider the simplest charging process via the diagonal two-level system (TLS) characterize by the probabilities $(s_1, s_2)$. For this process the unitary $U = u^{(1)} \oplus u^{(2)} \oplus \dots $ is composed of a trivial unit matrix $u^{(1)} = 1$ and two-dimensional matrices $u^{(k)}$ for $k>1$, such that the transition matrix simplifies to:
\begin{eqnarray}
T_{km} = \sum_{n=1}^\infty \sum_{i,j = 1}^{\min(n,d)}  |u^{(n)}_{ij}|^2 s_j \ \delta_{n-j+1,m} \ \delta_{n-i+1,k} = s_1  \delta_{1,m}  \delta_{1,k} + \sum_{n=2}^\infty \sum_{i,j = 1}^{2}  |u^{(n)}_{ij}|^2 s_j \ \delta_{n-j+1,m} \ \delta_{n-i+1,k}
\end{eqnarray}
Moreover, from the bistochasticity we put:
\begin{align}
   & |u_{11}^{(n)}|^2 = |u_{22}^{(n)}|^2 \equiv 1 - \alpha_n \\
   & |u_{12}^{(n)}|^2 = |u_{21}^{(n)}|^2 \equiv \alpha_n
\end{align}
for $n>1$ and where $\alpha_n \in [0,1]$. Using this, we get:
\begin{eqnarray}
   \sum_{i,j = 1}^{2}  |u^{(n)}_{ij}|^2 s_j  \delta_{n-j+1,m} \delta_{n-i+1,k} = 
   (1 - \alpha_n) (s_1  \delta_{n,m} \delta_{n,k} + s_2  \delta_{n-1,m}  \delta_{n-1,k}) + \alpha_n (s_2  \delta_{n-1,m}  \delta_{n,k} + s_1  \delta_{n,m}  \delta_{n-1,k})
\end{eqnarray}
and finally we obtain the following transition probabilities:
\begin{eqnarray} \label{qubit_tran_matrix}
      T_{km} = 
\begin{cases}
    (1 - \alpha_2 s_2) \delta_{k,1} +  \alpha_2 s_2 \delta_{k,2}, & m=1\\
    (1 - \alpha_{m} s_1 - \alpha_{m+1} s_2)  \delta_{k,m} + \alpha_{m} s_1 \delta_{k,m-1} + \alpha_{m+1} s_2  \delta_{k,m+1}, & m>1
\end{cases}
\end{eqnarray}

\subsection{Condition for the recurrent chain\label{rec}}
\begin{theorem}
The irreducible Markov chain with the transition matrix $T$ given by Eq. \eqref{qubit_tran_matrix} is recurrent if $s_1 \ge s_2$ (i.e., the state is passive). 
\end{theorem}
\begin{proof}
Let us assume that $m>1$, such that for arbitrary function $f_k$ we have
\begin{eqnarray}
   \sum_k T_{km} f_k = (1 - \alpha_{m} s_1 - \alpha_{m+1} s_2) f_m + \alpha_{m} s_1 f_{m-1} + \alpha_{m+1} s_2  f_{m+1}
\end{eqnarray}
and, as a consequence, the condition \eqref{function_inequality} simplifies to:
\begin{eqnarray}
  \alpha_{m} s_1 f_{m-1} + \alpha_{m+1} s_2  f_{m+1} \le (\alpha_{m} s_1 + \alpha_{m+1} s_2) f_m.
\end{eqnarray}
Then, let us put an abbreviation $f_{n+1} = f_n + \delta_n$, such that 
\begin{eqnarray} 
   \alpha_{m} s_1 (f_m - \delta_{m-1}) + \alpha_{m+1} s_2  (f_{m} + \delta_m) \le (\alpha_{m} s_1 + \alpha_{m+1} s_2) f_m
\end{eqnarray}
what gives us (for $m>1$)
\begin{eqnarray} 
   \alpha_{m+1} s_2  \delta_m \le \alpha_{m} s_1 \delta_{m-1}
   \end{eqnarray}
or
\begin{eqnarray}
   \delta_m \le \frac{\alpha_{m} s_1}{\alpha_{m+1} s_2} \delta_{m-1} \equiv g_{m-1} \delta_{m-1}.
\end{eqnarray} 
This further implies the following inequality:
\begin{eqnarray} \label{delta_condition}
   \delta_m \le g_{m-1} g_{m-2} \dots g_2 g_1 \delta_1
\end{eqnarray}
where 
\begin{eqnarray}
   g_{m-1} g_{m-2} \dots g_1 = \frac{\alpha_{m} s_1}{\alpha_{m+1} s_2} \frac{\alpha_{m-1} s_1}{\alpha_{m} s_2} \frac{\alpha_{m-2} s_1}{\alpha_{m-1} s_2} \dots \frac{\alpha_{3} s_1}{\alpha_{4} s_2} \frac{\alpha_{2} s_1}{\alpha_{3} s_2} = \frac{\alpha_{2}}{\alpha_{m+1}} \left(\frac{s_1}{s_2}\right)^{m-1}.
\end{eqnarray}
On the other hand, we have:
\begin{eqnarray}
   f_n = f_{n-1} + \delta_{n-1} = f_{n-2} + \delta_{n-2} +  \delta_{n-1} = \dots = f_1 + \sum_{k=1}^{n-1} \delta_k 
\end{eqnarray}
such that putting Eq. \eqref{delta_condition}:
\begin{eqnarray}
   f_n = f_1 + \delta_1 + \sum_{k=2}^{n-1} \delta_k \le f_1 + \delta_1 + \delta_1 \alpha_{2} \sum_{k=2}^{n-1} \frac{1}{\alpha_{k+1}} \left(\frac{s_1}{s_2}\right)^{k-1}.
\end{eqnarray}
We see that for a passive state with $s_1 \ge s_2$, when $n \to \infty$, we have 
\begin{eqnarray}
\sum_{k=2}^{\infty} \frac{1}{\alpha_{k+1}} \left(\frac{s_1}{s_2}\right)^{k-1} \ge  \sum_{k=2}^{\infty} \left(\frac{s_1}{s_2}\right)^{k-1} \to \infty
\end{eqnarray}
since $\alpha_k \le 1$.

Finally, let us consider a function $f_n$ in the form (with $f_1 = 0$ and $\delta_1 = 1$ for simplicity):
\begin{eqnarray}
   f_n = 1 + \alpha_2 \sum_{k=2}^{n-1} \frac{1}{\alpha_{k+1}} \left(\frac{s_1}{s_2}\right)^{k-1}.
\end{eqnarray}
It is seen that for $m>1$ it obeys condition \eqref{function_inequality} and if $s_1 \ge s_2$ it is divergent for $n \to \infty$. Taking as a subset $A = \{1 \}$ (the vacuum state) we see that the chain with transition matrix \eqref{qubit_tran_matrix}, according to Proposition \ref{reccurent_transient_conditions}, is recurrent.
\end{proof}

\subsection{Stationary state}
We try to find the stationary solution for the charging process, such that
\begin{eqnarray}
   p_k = \sum_{m=1}^\infty T_{km} p_m.
\end{eqnarray}
Putting Eq. \eqref{qubit_tran_matrix}, we get the set of equations:
\begin{align}
   p_1 &= (1 - \alpha_2 s_2) p_1 + \alpha_{2} s_1 p_2, \\
   p_2 &=  \alpha_2 s_2 p_1 + (1 - \alpha_{2} s_1 - \alpha_{3} s_2) p_2 + \alpha_{3} s_1 p_3, \\
   p_3 &= \alpha_{3} s_2 p_2 + (1 - \alpha_{3} s_1 - \alpha_{4} s_2)  p_3 + \alpha_{4} s_1 p_4,  \\
   & \dots \\
   p_k &= \alpha_{k} s_2  p_{k-1} + (1 - \alpha_{k} s_1 - \alpha_{k+1} s_2)  p_k + \alpha_{k+1} s_1 p_{k+1},
\end{align}
or equivalently
\begin{align}
&- \alpha_2 s_2 p_1 + \alpha_{2} s_1 p_2 = 0,\\
&\alpha_2 s_2 p_1  - (\alpha_{2} s_1 + \alpha_{3} s_2) p_2 + \alpha_{3} s_1 p_3 =0, \\
&\alpha_{3} s_2 p_2 - (\alpha_{3} s_1 + \alpha_{4} s_2)  p_3 + \alpha_{4} s_1 p_4  =0,\\
& \dots \\
&\alpha_{k} s_2  p_{k-1} - (\alpha_{k} s_1 + \alpha_{k+1} s_2)  p_k + \alpha_{k+1} s_1 p_{k+1} =0.
\end{align}
From this one can easily observe that the non-trivial solution is for probabilities obeying the following relation:
\begin{eqnarray} \label{gibbs_condition}
   p_k = \frac{s_1}{s_2} p_{k+1}.
\end{eqnarray}
However, since we looking for the normalized set of probabilities, i.e., $\sum_k p_k =1$, this can only be satisfied if $s_1 > s_2$, i.e., the stationary state exist only for strictly passive states.

\begin{remark}
The stationary state obeying the condition \eqref{gibbs_condition} is the Gibbs state. Notice that for repeating charging via the strictly passive state (with $s_1 > s_2$) the battery approaches the Gibbs state, which is independent of the applied unitary protocol (parameterized by the coefficients $\alpha_k$). 
\end{remark}

\subsection{Transient chain}
Exploiting Proposition \ref{reccurent_transient_conditions}, we show below that for any TLS in an active state, one can always find a unitary which results in a transient Markov chain. To prove that, we consider a particular unitary in the form (see Eq. \eqref{unitary_example} in the main text): 
\begin{eqnarray} 
    U =  \dyad{1,g} + \sum_{n=2}^\infty \big(\dyad{n,g}{n-1,e} + \dyad{n-1,e}{n,g} \big).
\end{eqnarray}
By direct calculation, we obtain:
\begin{align}
    \Tr_S[U \rho \otimes \xi U^\dag] 
   &= \dyad{1}{1}\Big(p_{1}+p_{2}\Big)s_{1}+ \sum_{n=2}^\infty \dyad{n}{n}   \Big(p_{n-1} s_2+p_{n+1}s_{1}\Big).  
\end{align}
Consequently, the inequality (\ref{function_inequality}) for all
$m \in \mathbb{N} \setminus A$, with $A=\{1\}$, takes the form
\begin{align}\label{transient_condition}
    s_{1}f_{m-1}+s_{2}f_{m+1}\leq f_{m}.
\end{align}
Similarly to the proof of the recurrence (section \ref{rec}), we choose $f_{k}=f_{k-1}-\delta_{k-1}$, $f_{k+1}=f_{k}-\delta_{k}=f_{k-1}-\delta_{k}-\delta_{k-1}$ with positive $\delta$ factors. This guarantees that the funtion $f_{k}$ decreses with $k$, which implies that the condition
$f_k < \inf_{m \in A} f_m$ for at least one $k \in \mathbb{N} \setminus A$ is satisfied. We also select $\delta_k$ such that the positivity condition $f_{k}>0$ is also satisfied. This is achieved by taking $\delta_{k+1}=a\cdot\delta_{k}$, with $0<a<1$. Then, we have $f_{k}=f_{1}-\sum_{i=1}^{k-1}\delta_{i}=f_{1}-\delta_{1}\sum_{i=0}^{k-1}a^{i}$. The series $\sum_{i=0}^{\infty}a^{i}$ converges for arbitrary $0<a<1$, and therefore it is always possible to set $f_{1}$ large enough such that $f_{k}>0$ for arbitrary $k$,  $\delta_{1}$ and $0<a<1$.

(\ref{transient_condition}) may be rewritten as
\begin{align}
    s_{1}f_{m-1}+s_{2}(f_{m-1}-\delta_{m-1}-\delta_{m})\leq f_{m-1}-\delta_{m-1}, 
\end{align}
and since $s_{1}+s_{2}=1$,
\begin{align}
    s_{2}\geq \frac{\delta_{m-1}}{\delta_{m}+\delta_{m-1}}=\frac{1}{1+a}.
\end{align}
By approaching with $a$ the value 1 from the left side: $a\rightarrow 1^-$, we obtain $s_{2}>\frac{1}{2}$, and therefore prove that active qubit states are related to transient chains.


\end{document}